\documentclass[11pt]{article}
\usepackage[margin=1in]{geometry}
\usepackage[utf8]{inputenc}
\usepackage{microtype}
\usepackage[all=normal,paragraphs=tight,bibliography=tight,floats=tight]{savetrees}
\usepackage{listings}
  \usepackage{mathrsfs}
\usepackage{amsfonts,amsthm,amsmath}
\usepackage{amstext}
\usepackage{graphicx,tikz}
\usepackage{comment}
\usepackage{url}

\newtheorem{theorem}{Theorem}
\newtheorem{lemma}[theorem]{Lemma}

\theoremstyle{definition}

\newcommand{\olddaniel}[1]{}

\usepackage{listings}

\newcommand{\bran}[1]{branchable\xspace}

\newtheorem{claim}{Claim}

\usepackage{amsmath, amssymb, latexsym}
\usepackage{enumerate}

\usepackage{xspace}

\newcommand{\Oh}{{\mathcal{O}}}

\def\cqedsymbol{\ifmmode$\lrcorner$\else{\unskip\nobreak\hfil
\penalty50\hskip1em\null\nobreak\hfil$\lrcorner$
\parfillskip=0pt\finalhyphendemerits=0\endgraf}\fi} 

\newcommand{\cqed}{\renewcommand{\qed}{\cqedsymbol}}

\newcommand{\eps}{\epsilon}

\title{Below all subsets for Minimal Connected Dominating Set\thanks{
The research of S. Saurabh leading to these results has received funding from the European Research Council under the European Union's Seventh Framework Programme (FP/2007-2013) / ERC Grant Agreement no.~306992.
The research of Mi. Pilipczuk is supported by Polish National Science Centre grant UMO-2013/11/D/ST6/03073.
Mi. Pilipczuk is also supported by the Foundation for Polish Science via the START stipend programme.}}
\author{
  Daniel Lokshtanov\thanks{
    Department of Informatics, University of Bergen, Norway, \texttt{daniello@ii.uib.no}.
  }
  \and
  Micha\l{} Pilipczuk\thanks{
    Institute of Informatics, University of Warsaw, Poland, \texttt{michal.pilipczuk@mimuw.edu.pl}.
  }
  \and 
  Saket Saurabh\thanks{
    Institute of Mathematical Sciences, India, \texttt{saket@imsc.res.in}, and
    Department of Informatics, University of Bergen, Norway, \texttt{Saket.Saurabh@ii.uib.no}.
  }
}

\date{}

\begin{document}

\maketitle

\begin{abstract}
A vertex subset $S$ in a graph $G$ is a {\em dominating set} if every vertex not contained in $S$ has a neighbor in $S$. 
A dominating set $S$ is a {\em connected dominating set} if the subgraph $G[S]$ induced by $S$ is connected. 
A connected dominating set $S$ is a {\em minimal connected dominating set} if no proper subset of $S$ is also a connected dominating set. 
We prove that there exists a constant $\eps > 10^{-50}$ such that every graph $G$ on $n$ vertices has at most $\Oh(2^{(1-\eps)n})$ minimal connected dominating sets. 
For the same $\eps$ we also give an algorithm with running time $2^{(1-\eps)n}\cdot n^{\Oh(1)}$ to enumerate all minimal connected dominating sets in an input graph $G$.

\end{abstract}


\section{Introduction}

In the field of {\em{enumeration}} algorithms, the following setting is commonly considered.
Suppose we have some universe $U$ and some property $\Pi$ of subsets of $U$.
For instance, $U$ can be the vertex set of a graph $G$, whereas $\Pi$ may be the property of being an independent set in $G$, or a dominating set of $G$, etc.
Let $\mathcal{F}$ be the family of all {\em{solutions}}: subsets of $U$ satisfying~$\Pi$.
Then we would like to find an algorithm that enumerates all solutions quickly,
optimally in time $|\mathcal{F}|\cdot n^{\Oh(1)}$, where $n$ is the size of the universe.
Such an enumeration algorithm could be then used as a subroutine for more general problems.
For instance, if one looks for an independent set of maximum possible weight is a vertex-weighted graph,
if suffices to browse through all inclusion-wise maximal independent sets (disregarding the weights) and pick the one with the largest weight.
If maximal independent sets can be efficiently enumerated and their number turns out to be small, such an algorithm could be useful in practice.

The other motivation for enumeration algorithms stems from extremal problems for graph properties.
Suppose we would like to know what is, say, the maximum possible number of inclusion-wise maximal independent sets in a graph on $n$ vertices.
Then it suffices to find an enumeration algorithm for maximal independent sets, and bound its (exponential) running time in terms of $n$.
The standard approach for the design of such an enumeration algorithm is to construct a smart branching procedure.
The run of such a branching procedure can be viewed as a tree where the nodes correspond to moments when the algorithm branches into two or more subprocedures, fixing different choices for the shape of a solution.
Then the leaves of such a search tree correspond to the discovered solutions.
By devising smart branching rules one can limit the number of leaves of the search tree, which both estimates the running time of the enumeration algorithm, and provides a combinatorial upper bound
on the number of solution. 
For instance, the classic proof of Moon and Moser~\cite{moon-moser} that the number of maximal independent sets in an $n$-vertex graph is at most $3^{n/3}$,
can be easily turned into an algorithm enumerating this family in time $3^{n/3}\cdot n^{\Oh(1)}$.

However, the analysis of branching algorithms is often quite nontrivial. The technique usually used, called {\em{Measure\&Conquer}}, involves assigning auxiliary potential measures to
 subinstances obtained during branching, and analyzing how the potentials change during performing the branching rules.
Perhaps the most well-known result obtained using Measure\&Conquer is the $\Oh(1.7159^n)$-time algorithm of Fomin et al.~\cite{FominGPS08} for enumerating minimal dominating sets. Note that
in particular this implies an $\Oh(1.7159^n)$ upper bound on the number of minimal dominating sets.
We refer to the book of Fomin and Kratsch~\cite{FominK10} for a broader discussion of branching algorithms and the Measure\&Conquer technique.

The main limitation of such branching strategies is that, without any closer insight, they can only handle properties that are somehow local.
This is because pruning unnecessary branches is usually done by analyzing specific local configurations in the graph.
For this reason, it is difficult to add requirements of global nature to the framework.
To give an example, consider the family of minimal connected dominating sets of a graph: a subset of vertices $S$ is a {\em{minimal connected dominating set}} if it induces a connected subgraph, is a dominating
set, and no its proper subset has both these properties. While the number of minimal dominating sets of an $n$-vertex graph is bounded by $\Oh(1.7159^n)$ by the result of Fomin et al.~\cite{FominGPS08},
for the number of minimal connected dominating sets no upper bound of the form $\Oh(c^n)$ for any $c<2$ was known prior to this work.
The question about the existence of such an upper bound was asked by Golovach et al.~\cite{GolovachHK16}, 
and then re-iterated by Kratsch~\cite{report-enumeration} during the recent Lorentz workshop ``Enumeration Algorithms using Structure''.
We remark that the problem of finding a minimum-size connected dominating set admits an algorithm with running time $\Oh(1.9407^n)$~\cite{FominGK08}, 
but the method does not generalize to enumerating minimal connected dominating sets.

\subparagraph*{Our results.} We resolve this question in affirmative by proving the following theorem.

\begin{theorem}\label{thm:main} 
There is a constant $\eps>10^{-50}$ such that every graph $G$ on $n$ vertices has at most $\Oh(2^{(1-\eps)n})$ minimal connected dominating sets. 
Further, there is an algorithm that given as input a graph $G$, lists all minimal connected dominating sets of $G$ in time $2^{(1-\eps)n}\cdot n^{\Oh(1)}$.
\end{theorem}

Note that we give not only an improved combinatorial upper bound, but also a corresponding enumeration algorithm.
The improvement is minuscule, however our main motivation was just to break the trivial $2^n$ upper bound of enumerating all subsets.
In many places our argumentation could be improved to yield a slightly better bound at the cost of more involved analysis.
We choose not to do it, as we prefer to keep the reasoning as simple as possible, while the improvements would not decrease our upper bound drastically anyway.
The main purpose of this work is to show the possibility of achieving an upper bound exponentially smaller than $2^n$, and thus to investigate what tools could be useful for the treatment
of requirements of global nature in the setting of extremal problems for graph properties.

For the proof of Theorem~\ref{thm:main}, clearly it is sufficient to bound the number of minimal connected dominating sets of size roughly $n/2$. 
The starting point is the realization that any vertex $u$ in a minimal connected dominating set $S$ serves one of two possible roles. First, $u$ can be essential
for domination, which means that there is some $v$ not in $S$ such that $u$ is the only neighbor of $v$ in $S$. Second, $u$ can be essential for connectivity, in the sense that after removing $u$, the
subgraph induced by $S$ would become disconnected. Therefore, if we suppose that the vertices essential for domination form a small fraction of $S$, 
we infer that almost every vertex of $G[S]$ is a cut-vertex of this graph. 
It is not hard to convince oneself that then almost every vertex of $S$ has degree at most $2$ in $G[S]$. 

All in all, regardless whether the number of vertices essential for domination is small or large, a large fraction of all the vertices of the graph has at most $2$ neighbors in $S$.
Intuitively, in an ``ordinary'' graph the number of sets $S$ with this property should be significantly smaller than $2^n$.
We prove that this is indeed the case whenever the graph is ``robustly dense'' in the following sense: 
it has a spanning subgraph where almost all vertices have degrees not smaller than some constant $\ell$, but no vertex has degree larger than some (much larger) constant $h$.
For this proof we use the probabilistic method: we show that if $S$ is sampled at random, then the probability that many vertices are adjacent to at most $2$ vertices of $S$ is exponentially small.
The main tool is Chernoff-like concentration of independent random variables.

The remaining case is when the spanning subgraph as described above cannot be found. We attempt at constructing it using a greedy procedure, which in case of failure discovers a different
structure in the graph. We next show that such a structure can be also used to design an algorithm for enumerating minimal connected dominating sets faster than $2^n$, using a more direct
branching strategy. The multiple trade-offs made in this part of the proof are the main reason for why our improvement over the trivial $2^n$ upper bound is so small.

\section{Preliminaries}\label{sec:prelims}
All graphs considered in this paper are simple, i.e., they do not have self-loops or multiple edges connecting the same pair of vertices.
For a graph $G$, by $V(G)$ and $E(G)$ we denote the vertex and edge sets of $G$, respectively.
The {\em{neighborhood}} of a vertex $v$ in a graph $G$ is denoted by $N_G(v)$, and consists of vertices adjacent to $v$.
The {\em{degree}} of $v$, denoted by $d(v)$, is defined the cardinality of its neighborhood. 
For a subset $S \subseteq V(G)$ and vertex $v \in V(G)$ the {\em{$S$-degree}} of $v$, denoted $d_S(v)$, 
is defined to be the number of vertices in $S$ adjacent to $v$. 
A {\em proper coloring} of a graph $G$ with $c$ colors is a function $\phi \colon V(G) \rightarrow \{1, \ldots, c\}$ such that for every edge $uv \in E(G)$ we have $\phi(u) \neq \phi(v)$. 
For a proper coloring $\phi$ of $G$ and integer $i \leq c$, the $i$-th {\em color class} of $\phi$ is the set $V_i=\phi^{-1}(i)$.  
The subgraph of $G$ induced by a vertex subset $S \subseteq V(G)$ is denoted by $G[S]$ and defined to be the graph with vertex set $S$ and edge set $\{uv \in E(G) \colon u,v \in S\}$. 
For a vertex $v \in V(G)$, the graph $G - v$ is simply $G[V(G) \setminus \{v\}]$. 
A subset $I$ of vertices is {\em{independent}} if it induced an {\em{edgeless graph}}, that is, a graph with no edges.
A {\em cutvertex} in a connected graph $G$ is a vertex $v$ such that $G-v$ is disconnected.

We denote $\exp(t)=e^t$.
The probability of an event $A$ is denoted by $\Pr[A]$ and the expected value of a random variable $X$ is denoted by $E[X]$.
We use standard concentration bounds for sums of independent random variables.
In particular, the following variant of the Hoeffding's bound~\cite{H63}, given by Grimmett and Stirzaker \cite[p. 476]{GSbook}, will be used.  
\begin{theorem}[Hoeffding's bound]\label{prop:hoeff}
Suppose $X_1, X_2, \ldots, X_n$ are independent random variables such that $a_i \leq X_i \leq b_i$ for all $i$. 
Let $X = \Sigma_{i=1}^n X_i$. Then: 
$$\Pr[X-E[X] \geq t] \leq \exp\left(\frac{-2t^2}{\Sigma_{i=1}^n \left(b_i-a_i\right)^2}\right).$$
\end{theorem}

For enumeration, we need the following folklore claim.

\begin{lemma}\label{lem:enumeration}
Let $U$ be a universe of size $n$ and let $\mathcal{F}\subseteq 2^U$ be a family of its subsets that is closed under taking subsets ($X\subseteq Y$ and $Y\in \mathcal{F}$ implies $X\in \mathcal{F}$),
and given a set $X$ it can be decided in polynomial time whether $X\in \mathcal{F}$. Then $\mathcal{F}$ can be enumerated in time $|\mathcal{F}|\cdot n^{\Oh(1)}$.
\end{lemma}
\begin{proof}
Order the elements of $U$ arbitrarily as $e_1,e_2,\ldots,e_n$, and process them in this order while keeping some set $X\in \mathcal{F}$, initially set to be the empty set.
When considering the next $e_i$, check if $X\cup \{e_i\}\in \mathcal{F}$. If this is not the case, just proceed further with $X$ kept.
Otherwise, output $X\cup \{e_i\}$ as the next discovered set from $\mathcal{F}$, and execute two subprocedures: in the first proceed with $X$, and in the second proceed with $X\cup \{e_i\}$.
It can be easily seen that every set of $\mathcal{F}$ is discovered by the procedure, and that some new set of $\mathcal{F}$ is always discovered within a polynomial number of steps (i.e., this is
a polynomial-delay enumeration algorithm). Thus, the total running time is $|\mathcal{F}|\cdot n^{\Oh(1)}$.
\end{proof}

\section{Main case distinction}\label{sec:spanning}
The first step in our proof is to try to find a spanning subgraph of the considered graph $G$, which has constant maximum degree, 
but where only a small fraction of vertices have really small degrees. This is done by performing a greedy construction procedure.
Obviously, such a spanning subgraph may not exist, but then we argue that the procedure uncovers some other structure in the graph, which may be exploited by other means.
The form of the output of the greedy procedure constitutes the main case distinction in our proof.

\begin{lemma}\label{lem:subgraphOrPartition}
There is an algorithm that given as input a graph $G$, together with integers $\ell$ and $h$ such that $1 \leq \ell \leq h$, and a real $\delta$ with $0 \leq \delta \leq 1$, runs in polynomial time and outputs one of the following two objects:
\begin{enumerate}
\item\label{case:subgraph} A subgraph $G'$ of $G$ with $V(G') = V(G)$, such that 
\begin{itemize}
\item every vertex in $G'$ has degree at most~$h$, and 
\item less than $\delta \cdot n$ vertices in $G'$ have degree less than $\ell$.
\end{itemize}
\item\label{case:partition} A partition of $V(G)$ into subsets $L$, $H$ and $R$ such that 
\begin{itemize}
\item $|L| \geq \delta \cdot n$, 
\item every vertex in $L$ has strictly less than $\ell$ neighbors outside $H$, and 
\item $|H| \leq \frac{2\ell}{h} \cdot n$.
\end{itemize}
\end{enumerate}
\end{lemma}
\begin{proof}
The algorithm takes as input $\ell$, $h$ and $\delta$ and computes a subgraph $G'$ of $G$ as follows. Initially $V(G') = V(G)$ and $E(G') = \emptyset$. As long as there is an edge $uv \in E(G) \setminus E(G')$ such that (a) both $u$ and $v$ have degree strictly less than $h$ in $G'$, and (b) at least one of $u$ and $v$ has degree strictly less than $\ell$ in $G'$, the algorithm adds the edge $uv$ to $E(G')$. When the algorithm terminates, $G'$ is a subgraph of $G$ with $V(G') = V(G)$, such that every vertex in~$G'$ has degree at most $h$. Let $L$ be the set of vertices that have degree {\em strictly less} than $\ell$ in~$G'$. If $|L| < \delta \cdot n$ then the algorithm outputs $G'$, as $G'$ satisfies the conditions of case~\ref{case:subgraph}.

Suppose now that $|L| \geq \delta \cdot n$. Let $H$ be the set of vertices of degree exactly $h$ in $G'$, and let $R$ be $V(G) \setminus (L \cup H)$. Clearly $L$, $H$, and $R$ form a partition of $V(G)$. Consider any vertex $u \in L$. There can not exist an edge $uv \in E(G) \setminus E(G')$ with $v \notin H$, since such an edge would be added to $E(G')$ by the algorithm. Thus every vertex $v \in N_G(u) \setminus H$ is also a neighbor of $u$ in $G'$. Since the degree of $u$ in $G'$ is less than $\ell$, we conclude that $|N_G(u) \setminus H| < \ell$.

Finally, we show that $|H| \leq \frac{2\ell}{h} \cdot n$. To that end, we first upper bound $|E(G')|$. Consider the potential function 
\begin{equation*}
\phi(G') = \sum_{v \in V(G')} \max(\ell - d_{G'}(v), 0).
\end{equation*}
Initially the potential function has value $n\ell$. Each time an edge is added to $G'$ by the algorithm, the potential function decreases by (at least) $1$, because at least one endpoint of the added edge has degree less than $\ell$. Further, when the potential function is $0$, there are no vertices of degree less than $\ell$, and so the algorithm terminates. Thus, the algorithm terminates after at most $n\ell$ iterations, yielding $|E(G')| \leq n\ell$. Hence, the sum of the degrees of all vertices in $G'$ is at most $2n\ell$. Since every vertex in $H$ has degree $h$, it follows that~$|H| \leq \frac{2\ell}{h} \cdot n$.
\end{proof}

To prove Theorem~\ref{thm:main}, we apply Lemma~\ref{lem:subgraphOrPartition} with $\ell = 14$, $h = 3 \cdot 10^5$ and $\delta = \frac{1}{60}$. 
There are two possible outcomes. In the first case we obtain a subgraph $G'$ of $G$ with $V(G') = V(G)$, such that every vertex in $G'$ has degree at most $3 \cdot 10^5$, 
and at most $\frac{1}{60} \cdot n$ vertices in $G'$ have degree less than $14$. We handle this case using the following Lemma~\ref{lem:mainProbabilisticCase}, proved in Section~\ref{sec:probabilisticBound}.
\begin{lemma}\label{lem:mainProbabilisticCase}
Let $G$ be a graph on $n$ vertices that has a subgraph $G'$ with $V(G')=V(G)$ and the following properties: every vertex in $G'$ has degree at most $3 \cdot 10^5$, and less than $\frac{1}{60} \cdot n$ vertices in $G'$ have degree less than $14$. Then $G$ has at most $\Oh(2^{n \cdot \left(1 - 10^{-26} \right)})$ minimal connected dominating sets. Further, there is an algorithm that given as input $G$ and $G'$, enumerates the family of all minimal connected dominating sets of $G$ in time $2^{n \cdot \left(1 - 10^{-26} \right)}\cdot n^{\Oh(1)}$. 
\end{lemma}

In the second case we obtain a partition of $V(G)$ into $L$, $H$, and $R$ such that $|L| \geq \frac{1}{60} \cdot n$, every vertex in $L$ has strictly less than $14$ neighbors outside $H$, and $|H| \leq \frac{3}{10^4} \cdot n$.  This case is handled by the following Lemma~\ref{lem:mainStupid}, which we prove in Section~\ref{sec:stupidBranching}.
\begin{lemma}\label{lem:mainStupid}
Let $G$ be a graph on $n$ vertices that has a partition of $V(G)$ into $L$, $H$ and $R$ such that $|L| \geq \frac{1}{60} \cdot n$, every vertex in $L$ has strictly less than $14$ neighbors outside $H$, and $|H| \leq \frac{1}{10^4} \cdot n$. Then $G$ has at most $2^{n \cdot \left(1 - 10^{-50} \right)}$ minimal connected dominating sets. Further, there is an algorithm that given as input $G$ together with the partition $(L,H,R)$, enumerates the family of all minimal connected dominating sets of $G$ in time $2^{n \cdot \left(1 - 10^{-50} \right)}\cdot n^{\Oh(1)}$. 
\end{lemma}

Together, Lemmas~\ref{lem:mainProbabilisticCase} and~\ref{lem:mainStupid} complete the proof of Theorem~\ref{thm:main}.

\section{Robustly dense graphs}\label{sec:probabilisticBound}
In this section we bound the number of minimal connected dominating sets in a graph $G$ that satisfies case~\ref{case:subgraph} of Lemma~\ref{lem:subgraphOrPartition}, that is, we prove Lemma~\ref{lem:mainProbabilisticCase}. In particular, we assume that $G$ has a subgraph $G'$ such that all vertices of $G'$ have degree at most $h=3 \cdot 10^5$, and less than $\delta n=\frac{1}{60} n$ vertices of $G'$ have degree less than $\ell=14$. For a set $S$, we say that a vertex $v$ has {\em low $S$-degree} if $d_S(v) \leq 2$. We define the set $L(S) = \{v \in V(G) ~:~ d_S(v) \leq 2\}$ to be the set of vertices in $G$ of low $S$-degree.  Our bound consists of two main parts. In the first part we give an upper bound on the number of sets $S$ in $G$ such that $|L(S)| \geq \frac{1}{20} \cdot n$. In the second part we show that for any minimal connected dominating set $S$ of $G$ of size at least $\frac{4}{10}n$, we have $|L(S)| \geq \frac{1}{20} \cdot n$. Together the two parts immediately yield an upper bound on the 
number of (and an enumeration algorithm for) minimal connected dominating sets in $G$. We begin by proving the first part using a probabilistic argument.

\begin{lemma}\label{lem:mainProbabilistic}
Let $H$ be a graph on $n$ vertices of maximum degree at most $h$, such that at most $\frac{1}{60} \cdot n$ vertices have degree less than $\ell\geq 14$. Then there are at most $h^2 \cdot 2^n \cdot e^{- \frac{n}{1800 h^4}}$ subsets $S$ of $V(H)$ such that $|L(S)| \geq \frac{1}{20} \cdot n$. 
\end{lemma}

\begin{proof}
To prove the lemma, it is sufficient to show that if $S \subseteq V(H)$ is selected uniformly at random, then the probability that $|L(S)|$ is at least $\frac{1}{20} \cdot n$ is upper bounded as follows.

\begin{align}\label{eqn:mainProb}
\Pr\left[|L(S)| > \frac{1}{20} \cdot n\right] \leq h^2 \cdot \exp\left(- \frac{n}{1800 h^4}\right)
\end{align}

Let $H^2$ be the graph constructed from $H$ by adding an edge between every pair of vertices in $H$ that share a common neighbor. Since $H$ has maximum degree at most~$h$, $H^2$ has maximum degree at most $h(h-1)\leq h^2-1$, and therefore $H^2$ can be properly colored with $h^2$ colors~\cite{Diestelbook}. Let $\phi \colon V(H) \to \{1, \ldots, h^2\}$ be a proper coloring of $H^2$, and let $V_1, V_2, \ldots, V_{h^2}$ be the color classes of $\phi$. Two vertices in the same color class of $\phi$ have empty intersection of neighborhoods in $H$. Thus, when $S \subseteq V(H)$ is picked at random, we have that $d_S(u)$ and $d_S(v)$ are independent random variables whenever $u$ and $v$ are in the same color class of $\phi$. 

Let $Q$ be the set of vertices in $G$ of degree at least $\ell$. We have that $|Q| \geq (1 - \frac{1}{60}) \cdot n$ by assumption. For each $i \leq h^2$ we set $V_i^Q = V_i \cap Q$. Next we upper bound, for each $i \leq h^2$, the probability that $|L(S) \cap V_i^Q| > \frac{1}{40h^2}\cdot n$. 
For every vertex $v \in V(H)$, define the indicator variable $X_v$ which is set to $1$ if $d_S(v) \leq 2$ and $X_v$ is set to $0$ otherwise. We have that 
\begin{align*}
\Pr[X_v = 1] = \frac{{d(v) \choose 0} + {d(v) \choose 1} + {d(v) \choose 2}}{2^{d(v)}}.
\end{align*}
The right hand side is non-increasing with increasing $d(v)$, so since we assumed $d(v) \geq \ell$, we have that $\Pr[X_v = 1] \leq \frac{{\ell \choose 0} + {\ell \choose 1} + {\ell \choose 2}}{2^{\ell}}\leq \frac{\ell^2}{2^{\ell}}$. In particular this holds for every $v \in Q$. Thus, for every $i \leq h^2$ we have that $|L(S) \cap V_i^Q| = \sum_{v \in V_i^Q} X_v$ --- that is, $|L(S) \cap V_i^Q|$ is a sum of $|V_i^Q|$ independent indicator variables, each taking value $1$ with probability at most $\frac{\ell^2}{2^{\ell}}$. Thus, Hoeffding's inequality (Theorem~\ref{prop:hoeff}) yields
\begin{align}\label{eqn:colorClassProb}
\Pr\left[|L(S) \cap V_i^Q| \geq \frac{\ell^2}{2^{\ell}} \cdot |V_i^Q| + \frac{n}{60h^2}\right] \leq \exp\left(- \frac{2n^2}{3600h^4 |V_i^Q|}\right) \leq \exp\left(- \frac{n}{1800 h^4}\right).
\end{align}
The union bound over the $h^2$ color classes of $\phi$, coupled with equation~\eqref{eqn:colorClassProb}, yields that 
\begin{align*}
\Pr\left[|L(S) \cap Q| \geq \frac{\ell^2}{2^{\ell}} |Q| + \frac{1}{60} \cdot n\right] \leq h^2 \cdot \exp\left(- \frac{n}{1800 h^4}\right).
\end{align*}
Hence, with probability at least $1 - h^2 \cdot \exp\left(- \frac{n}{1800 h^4}\right)$ we have that 
\begin{align*}
|L(S) \cap Q| \leq \frac{\ell^2}{2^{\ell}} |Q| + \frac{1}{60} \cdot n \leq \frac{2}{60} \cdot n,
\end{align*}
where the last inequality holds due to $\ell \geq 14$. Since $|L(S)| \leq |L(S) \cap Q| + |V(H) \setminus Q|$ and $|V(H) \setminus Q| \leq \frac{1}{60}n$ it follows that in this case, $|L(S)| \leq \frac{1}{20} \cdot n$. This proves equation~\eqref{eqn:mainProb} and the statement of the Lemma.
\end{proof}

Note that the statement of Lemma~\ref{lem:mainProbabilistic} requires $H$ to have at maximum degree at most $h$, such that at most $\frac{1}{60} \cdot n$ vertices have degree less than $\ell$. What we obtain from Lemma~\ref{lem:subgraphOrPartition} is a subgraph $G'$ of the input graph $G$ with these properties. We will apply Lemma~\ref{lem:mainProbabilistic} to $H=G'$ and transfer the conclusion to $G$, since $G'$ is a subgraph of $G$.

We now turn to proving the second part, that for any minimal connected dominating set $S$ of $G$ of size at least $\frac{4}{10}n$, we have $|L(S)| \geq \frac{1}{20} \cdot n$. The first step of the proof is to show that any graph where almost every vertex is a cut vertex must have many vertices of degree $2$.

\begin{lemma}\label{lem:cutAndCount}
Let $\alpha>0$ be a constant. Suppose that $H$ is a connected graph on $n$ vertices in which at least $(1-\alpha)n$ vertices are cutvertices.
Then at least $(1-7\alpha)n$ vertices of $H$ have degree equal to $2$.
\end{lemma}
\begin{proof}
Let $X$ be the set of those vertices of $H$ that are not cutvertices.
By the assumption we have $|X|\leq \alpha n$.
Let $T$ be any spanning tree in $H$, and let $L_1$ be the set of leaves of $T$.
No leaf of $T$ is a cutvertex of $H$, hence $L_1\subseteq X$.
Let $L_3$ be the set of those vertices of $T$ that have degree at least $3$ in $T$.
It is well-known that in any tree, the number of vertices of degree at least $3$ is smaller than the number of leaves.
Therefore, we have the following:
\begin{equation}\label{eq:l1l3}
|L_3|<|L_1|\leq |X|\leq \alpha n.
\end{equation}

Let $R$ be the closed neighborhood of $L_1\cup L_3\cup X$ in $T$, that is, the set consisting of $L_1\cup L_3\cup X$ and all vertices that have neighbors in $L_1\cup L_3\cup X$.
Since $T$ is a tree, it can be decomposed into a set of paths $\mathcal{P}$, where each path connects two vertices of $L_1\cup L_3$ and all its internal vertices have degree $2$ in $T$.
Contracting each of these paths into a single edge yields a tree on the vertex set $L_1\cup L_3$, which means that the number of the paths in $\mathcal{P}$ is less than $|L_1\cup L_3|$.
Note that the closed neighborhood of $L_1\cup L_3$ in $T$ contains at most $2$ of the internal vertices on each of the paths from $\mathcal{P}$: the first and the last one.
Moreover, each vertex of $X\setminus (L_1\cup L_3)$ introduces at most $3$ vertices to $R$: itself, plus two its neighbors on the path from $\mathcal{P}$ on which it lies.
Consequently, by equation~\eqref{eq:l1l3} we have:
\begin{equation}\label{eq:R}
|R|\leq |L_1|+|L_3|+2|L_1\cup L_3|+3|X|\leq 7\alpha n.
\end{equation}

We now claim that every vertex $u$ that does not belong to $R$, in fact has degree $2$ in~$H$.
By the definition of $R$ we have that $u$ has degree $2$ in $T$, both its neighbors $v_1$ and $v_2$ in~$T$ also have degree $2$ in $T$, and moreover $u$, $v_1$, and $v_2$ are all cutvertices in $H$.
Aiming towards a contradiction, suppose $u$ has some other neighbor $w$ in $H$, different than $v_1$ and $v_2$. 
Then the unique path from $u$ to $w$ in $T$ passes either through $v_1$ or through $v_2$; say, through $v_1$.
However, then the removal of $v_1$ from $H$ would not result in disconnecting $H$.
This is because the removal of $v_1$ from $T$ breaks $T$ into $2$ connected components, as the degree of $v_1$ in $T$ is equal to $2$, 
and these connected components are adjacent in $H$ due to the existence of the edge $uw$. This is a contradiction with the assumption that $v_1$ and $v_2$ are cutvertices.

From equation~\eqref{eq:R} and the claim proved above it follows that at least $(1-7\alpha)n$ vertices of $G$ have degree equal to $2$.
\end{proof}

We now apply Lemma~\ref{lem:cutAndCount} to subgraphs induced by minimal connected dominating sets.

\begin{lemma}\label{lem:LSlowerbound}
Let $S$ be a minimal connected dominating set of a graph $G$ on $n$ vertices, such that $|S| \geq \frac{4}{10}n$. Then $|L(S)| \geq \frac{1}{20}n$.
\end{lemma}
\begin{proof}
For $n\leq 2$ the claim is trivial, so assume $n\geq 3$, which in particular implies $|S|\geq 2$.
Aiming towards a contradiction, suppose $|L(S)| < \frac{1}{20}n$. By minimality, we have that for every vertex $v$, the set $S \setminus \{v\}$ is not a connected dominating set of $G$. Let
\begin{equation*}
S_{\textrm{cut}} = \{v \in S\colon G[S] - v \mbox{ is disconnected}\}.
\end{equation*}
Consider a vertex $v$ in $S \setminus S_{\textrm{cut}}$. We have that $S \setminus \{v\}$ can not dominate all of $V(G)$ because otherwise $S \setminus \{v\}$ would be a connected dominating set.  Let $u$ be a vertex in of $G$ not dominated by $S\setminus \{v\}$. Because $G[S]$ is connected and $|S|\geq 2$, vertex $v$ has a neighbor in $S$, so in particular $u\neq v$ and hence $u\notin S$. Further, since $S$ is a connected dominating set, $u$ has a neighbor in $S$, and this neighbor can only be $v$. Hence $d_S(u) = 1$ and so $u \in L(S)$. Re-applying this argument for every $v \in S \setminus S_{\textrm{cut}}$ yields $|L(S)| \geq |S \setminus S_{\textrm{cut}}|$. 

From the argument above and the assumption $|L(S)<\frac{1}{20}n$, it follows that $|S \setminus S_{\textrm{cut}}| \leq \frac{1}{20}n$. Since $|S| \geq \frac{4}{10}n$, we have that $|S \setminus S_{\textrm{cut}}| \leq \frac{1}{8}|S|$. It follows that $|S_{\textrm{cut}}| \geq \frac{1}{8}|S|$. By Lemma~\ref{lem:cutAndCount} applied to $G[S]$, the number of degree $2$ vertices in $G[S]$ is at least $(1 - \frac{7}{8})|S| = \frac{1}{8}|S| \geq \frac{1}{20}n$. Each of these vertices belongs to $L(S)$, which yields the desired contradiction.
\end{proof}

We are now in position to wrap up the first case, giving a proof of Lemma~\ref{lem:mainProbabilisticCase}.

\begin{proof}[Proof of Lemma~\ref{lem:mainProbabilisticCase}]
There are at most $\sum_{i=0}^{\lfloor\frac{4n}{10}\rfloor}{n \choose i} \leq 2^{n(1-\frac{1}{100})}$ subsets of $V(G)$ of size at most $\frac{4}{10} \cdot n$. 
Thus, the family of all minimal connected dominating sets of size at most $\frac{4}{10} \cdot n$ can be enumerated in time $2^{n(1-\frac{1}{100})}\cdot n^{\Oh(1)}$ by enumerating all sets of size at most $\frac{4}{10} \cdot n$, and checking for each set in polynomial time whether it is a minimal connected dominating set.

Consider now any minimal connected dominating set $S$ in $G$ with $|S|\geq \frac{4}{10} \cdot n$. By Lemma~\ref{lem:LSlowerbound}, we have that $|L(S)| \geq \frac{1}{20}n$. Since every vertex of degree at most $2$ in $G$ has degree at most $2$ in $G'$, it follows that $|L(S)| \geq \frac{1}{20}n$ holds also in $G'$. However, by Lemma~\ref{lem:mainProbabilistic} applied to $G'$, there are at most $2^n \cdot e^{- \frac{n}{1800 h^4}}$ subsets $S$ of $V(G') = V(G)$ such that $|L(S)| \geq \frac{1}{20} \cdot n$ (in $G'$). Substituting $h = 3 \cdot 10^5$ in the above upper bound yields that there are at most $2^{n \cdot \left(1 -10^{-26} \right)}$ minimal connected dominating sets of size at least $\frac{4}{10}n$, yielding the claimed upper bound on the number of minimal connected dominating sets.

To enumerate all minimal connected dominating sets of $G$ of size at least $\frac{4}{10}n$ in time $2^{n \cdot \left(1 - 10^{-26} \right)}\cdot n^{\Oh(1)}$, it is sufficient to list all sets $S$ such that $|L(S)| \geq \frac{1}{20} \cdot n$, and for each such set determine in polynomial time whether it is a minimal connected dominating set. Note that the family of sets $S$ such that $|L(S)| \geq \frac{1}{20} \cdot n$ is closed under subsets: if $|L(S)| \geq \frac{1}{20} \cdot n$ and $S' \subseteq S$ then $|L(S')| \geq \frac{1}{20} \cdot n$. Since it can be tested in polynomial time for a set $S$ whether $|L(S)|\geq \frac{1}{20}\cdot n$, the family of all sets with $|L(S)| \geq \frac{1}{20} \cdot n$ can be enumerated in time  $2^{n \cdot \left(1 - 10^{-26} \right)}n^{\Oh(1)}$ by the algorithm of Lemma~\ref{lem:enumeration}, completing the proof.
\end{proof}

\section{Large sparse induced subgraph}\label{sec:stupidBranching}
In this section we bound the number of minimal connected dominating sets in any graph~$G$ for which case~\ref{case:partition} of Lemma~\ref{lem:subgraphOrPartition} occurs, i.e., we prove Lemma~\ref{lem:mainStupid}. Let us fix some integer $\ell \geq 1$.

Our enumeration algorithm will make decisions that some vertices are in the constructed connected dominating set, and some are not. We incorporate such decisions in the notion of {\em extensions}. For disjoint vertex sets $I$ and $O$ (for {\em in} and {\em out}), we define an $(I,O)$-{\em extension} to be a vertex set $S$ that is disjoint from $I \cup O$ and such that $I \cup S$ is a connected dominating set in $G$. An $(I,O)$-{\em extension} $S$ is said to be {\em minimal} if no proper subset of it is also an $(I,O)$-extension.
The following simple fact will be useful.

\begin{lemma}\label{lem:polyTimeExtension}
There is a polynomial-time algorithm that, given a graph $G$ and disjoint vertex subsets $I$, $O$, and $S$, determines whether $S$ is a minimal $(I, O)$-extension in $G$.
\end{lemma}
\begin{proof}
The algorithm checks whether $I \cup S$ is a connected dominating set in $G$ and returns ``no'' if not. Then, for each $v \in S$ the algorithm tests whether $I \cup (S \setminus \{v\})$ is a connected dominating set of $G$. If it is a connected dominating set for any choice of $v$, the algorithm returns ``no''. Otherwise, the algorithm returns that $S$ is a minimal $(I, O)$-extension. The algorithm clearly runs in polynomial time, and if the algorithm returns that $S$ is not a minimal $(I, O)$-extension in $G$, then this is correct, as the algorithm also provides a certificate.

We now prove that if $S$ is {\em not} a minimal $(I, O)$ extension in $G$, then the algorithm returns ``no''. If $S$ is not an $(I, O)$-extension at all, the algorithm detects it and reports no. If it is an $(I, O)$-extension, but not a minimal one, then there exists an $(I,O)$-extension $S' \subsetneq S$. Let $v$ be any vertex in $S \setminus S'$. We claim that $X = I \cup (S \setminus \{v\})$ is a connected dominating set of $G$. Indeed, $X$ dominates $V(G)$ because $I \cup S'$ does. Furthermore, $G[X]$ is connected because $G[I \cup S']$ is connected and every vertex in $X \setminus (I \cup S')$ has a neighbor in $(I \cup S')$. Hence $I \cup (S \setminus \{v\})$ is a connected dominating set of $G$ and the algorithm correctly reports ``no''. This concludes the proof.
\end{proof}

Observe that for any minimal connected dominating set $X$, and any $I \subseteq X$ and $O$ disjoint from $X$, we have that $X \setminus I$ is a minimal $(I,O)$-extension. Thus one can use an upper bound on the number of minimal extensions to upper bound the number of minimal connected dominating sets. Recall that case~\ref{case:partition} of Lemma~\ref{lem:subgraphOrPartition} provides us with a partition $(L,H,R)$ of the vertex set. To upper bound the number of minimal connected dominating sets, we will consider each of the $2^{n-|L|}$ possible partitions of $H \cup R$ into two sets $I$ and $O$, and upper bound the number of minimal $(I, O)$-extensions. This is expressed in the following lemma.

\begin{lemma}\label{lem:mainBranching}
Let $G$ be a graph and $(L,H,R)$ be a partition of the vertex set of $G$ such that $|L| \geq 10|H|\ell$, and every vertex in $L$ has less than $\ell$ neighbors in $L \cup R$. Then, for every partition $(I,O)$ of $H \cup R$, there are at most $2^{|L|} \cdot  e^{-\frac{|L|}{2^{-10\ell}\cdot 100\ell^3}}$ minimal $(I,O)$-extensions. Furthermore, all minimal $(I,O)$-extensions can be listed in time $2^{|L|} \cdot  e^{-\frac{|L|}{2^{10\ell}\cdot 100\ell^3}} \cdot n^{\Oh(1)}$.
\end{lemma}

We now prepare ground for the proof of Lemma~\ref{lem:mainBranching}. The first step is to reduce the problem essentially to the case when $L$ is independent.
For this, we shall say that a partition of $V(G)$ into $L$, $H$, and $R$ is a {\em good partition} if:
\begin{itemize}
\item $|L| \geq 10|H|$, 
\item $L$ is an independent set, and 
\item every vertex in $L$ has less than $\ell$ neighbors in $R$. 
\end{itemize}
Towards proving Lemma~\ref{lem:mainBranching}, we first prove the statement assuming that the input partition of $V(G)$ is a good partition.

\begin{lemma}\label{lem:mainBranchingIset}
Let $G$ be a graph and $(L,H,R)$ be a good partition of $V(G)$. Then, for every partition $(I,O)$ of $H \cup R$, there are at most $2^{|L|} \cdot  e^{-\frac{|L|}{2^{10\ell}\cdot 100\ell^2}}$ minimal $(I,O)$-extensions. Furthermore, all minimal $(I,O)$-extensions can be listed in time $2^{|L|} \cdot  e^{-\frac{|L|}{2^{10\ell}\cdot 100\ell^2}} \cdot n^{\Oh(1)}$.
\end{lemma}

We will prove Lemma~\ref{lem:mainBranchingIset} towards the end of this section, now let us first prove Lemma~\ref{lem:mainBranching} assuming the correctness of Lemma~\ref{lem:mainBranchingIset}.

\begin{proof}[Proof of Lemma~\ref{lem:mainBranching} assuming Lemma~\ref{lem:mainBranchingIset}] Observe that we may find an independent set $L'$ in $G[L]$ of size at least $\frac{|L|}{\ell}$. Indeed, since every vertex of $L$ has less than $\ell$ neighbors in $L\cup R$, any inclusion-wise maximal independent set $L'$ in $G[L]$ has size at least $\frac{|L|}{\ell}$. Therefore $|L'| \geq \frac{|L|}{\ell} \geq  10|H|$, and hence $L'$, $H$, $R' = R \cup (L' \setminus L)$ is a good partition of~$V(G)$.

Further, for a fixed partition of $R \cup H$ into $I$ and $O$, consider each of the $2^{|L \setminus L'|}$ partitions of $H \cup R'$ into $I'$ and $O'$ such that $I \subseteq I'$ and $O \subseteq O'$. For every minimal $(I,O)$-extension $S$, we have that $S \cap L'$ is a minimal $(I', O')$-extension, where $I' = I \cup (S \setminus L')$ and $O' = O \cup (L \setminus (L' \cup S))$. Thus, by Lemma~\ref{lem:mainBranchingIset} applied to the good partition $(L',H,R')$ of $V(G)$, and the partition $(I',O')$ of $H \cup R'$, we have that the number of minimal $(I,O)$-extensions is upper bounded by 
\begin{equation*}
2^{|L\setminus L'|}\cdot 2^{|L'|} \cdot  e^{-\frac{|L'|}{2^{10\ell}100\ell^2}} \leq 2^{|L|} \cdot  e^{-\frac{|L|}{2^{10\ell}100\ell^3}}.
\end{equation*}
Further, by the same argument, the minimal $(I,O)$-extensions can be enumerated within the claimed running time, using the enumeration provided by Lemma~\ref{lem:mainBranchingIset} as a subroutine.
   \end{proof}

The next step of the proof of Lemma~\ref{lem:mainBranchingIset} is to make a further reduction, this time to the case when also $H\cup R$ is independent.
Since the partition into vertices taken and excluded from the constructed connected dominating set is already fixed on $H\cup R$, this amounts to standard cleaning operations within $H\cup R$.
We shall say that a good partition $(L, H, R)$ of $V(G)$ is an {\em excellent partition} if $G[H \cup R]$ is edgeless.
 
\begin{lemma}\label{lem:reduceIO}
There exists an algorithm that given as input a graph $G$, together with a good partition $(L,R,H)$ of $V(G)$, and a partition $(I,O)$ of $R \cup H$, runs in polynomial time, and outputs a graph $G'$ with $V(G)\cap V(G')\supseteq L$, an excellent partition $(L,R',H')$ of $V(G')$, and a partition $(I',O')$ of $R' \cup H'$, with the following property. For every set $S \subseteq L$, $S$ is a minimal $(I,O)$-extension in $G$ if and only if $S$ is a minimal $(I', O')$-extension in $G'$.
\end{lemma}
\begin{proof}
The algorithm begins by setting $G' = G$, $H' = H$, $R' = R$, $I' = I$ and $O' = O$. It then proceeds to modify $G'$, at each step maintaining the following invariants: (i) $(L,H',R')$ is a good partition of the vertex set of $G'$, and (ii) for every set $S \subseteq L$, $S$ is a minimal $(I,O)$-extension in $G$ if and only if $S$ is a minimal $(I', O')$-extension in $G'$.

If there exists an edge $uv$ with $u \in O'$ and $v \in I'$, the algorithm removes $u$ from $G'$, from $O'$, and from $R'$ or $H'$ depending on which of the two sets it belongs to. Since $u$ is anyway dominated by $I'$ and removing $u$ can only decrease $|H'|$ (while keeping $|L|$ the same), the invariants are maintained.
If there exists an edge $uv$ with both $u$ and $v$ in $O'$, the algorithm removes the edge $uv$ from $G'$. Since neither $u$ nor $v$ are part of $I' \cup S$ for any $S \subseteq L$, it follows that the invariants are preserved.

Finally, if there exists an edge $uv$ with both $u$ and $v$ in $I'$, the algorithm contracts the edge $uv$. Let $w$ be the vertex resulting from the contraction. The algorithm removes $u$ and $v$ from $I'$ and from $R'$ or $H'$, depending on which of the two sets the vertices are in, and adds $w$ to $I'$. If at least one of $u$ and $v$ was in $H'$, $w$ is put into $H'$, otherwise $w$ is put into $R'$. Note that $|H'|$ may decrease, but can not increase in such a step. Thus $(L,R',H')$ remains a good partition and invariant (i) is preserved. Further, since $u$ and $v$ are always in the same connected component of $G'[I' \cup S]$ for any $S \subseteq L$, invariant (ii) is preserved as well. 

The algorithm proceeds performing one of the three steps above as long as there exists at least one edge in $G'[R' \cup H']$. When the algorithm terminates no such edge exists, thus $(L, H', R')$ forms an excellent partition of $V(G')$.
\end{proof}

Lemma~\ref{lem:reduceIO} essentially allows us to assume in the proof of Lemma~\ref{lem:mainBranchingIset} that $(L, H, R)$ is an excellent partition of $V(G)$. To complete the proof, we distinguish between two subcases: either there are at most $\frac{|L|}{10}$ vertices in $R$ of degree less than $10\ell$, or there are more than $\frac{|L|}{10}$ such vertices. 
Let us shortly explain the intuition behind this case distinction.
If there are at most $\frac{|L|}{10}$ vertices in $R$ of degree less than $10\ell$, then it is possible to show that $H \cup R$ is small compared to $L$, in particular that $|H \cup R| \leq \frac{3|L|}{10}$.
We then show that any minimal $(I,O)$-extension can not pick more than $|H \cup R|$ vertices from $L$. This gives a ${|L| \choose 0.3|L|}$ upper bound for the number of minimal $(I,O)$-extensions, which is significantly lower than $2^{|L|}$.

On the other hand, if there are more than $\frac{|L|}{10}$ vertices in $R$ of degree less than $10\ell$, then one can find a large subset $R'$ of $R$ of vertices of degree at most $10\ell$, such that no two vertices in $R'$ have a common neighbor. For each vertex $v \in R'$, every minimal $(I,O)$-extension must contain at least one neighbor of $v$. Thus, there are only $2^{d(v)}-1$, rather than $2^{d(v)}$ possibilities for how a minimal $(I,O)$-extension intersects the neighborhood of $v$. Since all vertices in $R'$ have disjoint neighborhoods, this gives an upper bound of  $2^{|L|} \cdot \left(\frac{2^{10\ell}-1}{2^{10\ell}}\right)^{|R'|}$ on the number of minimal $(I,O)$-extensions.

We now give a formal treatment of the two cases. We begin with the case that there are at most $\frac{|L|}{10}$ vertices in $R$ of degree less than $10\ell$.

\begin{lemma}\label{lem:smallSolution}
Let $G$ be a graph, and $I$ and $O$ be disjoint vertex sets such that $I$ is nonempty and both $G[I \cup O]$ and $G - (I \cup O)$ are edgeless. Then every minimal $(I,O)$-extension $S$ satisfies $|S| \leq |I \cup O|$.
\end{lemma}
\begin{proof}[Proof of Lemma~\ref{lem:smallSolution}]
We will need the following simple observation about the maximum size of an independent set of  internal nodes in a tree.
\begin{claim}\label{obs:internalIndependent} 
Let $T$ be a tree and $S$ be a set of non-leaf nodes of $T$ such that $S$ is independent in $T$. Then $|S| \leq |V(T) \setminus S|$.
\end{claim}
\begin{proof}
Root the tree $T$ at an arbitrary vertex. Construct a vertex set $Z$ by picking, for every $s \in S$, any child $z$ of $s$ and inserting $z$ into $Z$; this is possible since no vertex of $S$ is a leaf.
Every vertex in $T$ has a unique parent, so no vertex is inserted into $Z$ twice, and hence $|Z| = |S|$. Further, since $S$ is independent, $Z \subseteq V(T) \setminus S$. The claim follows.
\cqed\end{proof}

We proceed with the proof of the lemma.
Let $X = V(G) \setminus (I \cup O)$ and let $S \subseteq X$ be a minimal $(I,O)$-extension. Since $I \cup S$ is a connected dominating set and $I \cup O$ is independent, it follows that every vertex in $O$ has a neighbor in $S$. Hence $G[I \cup S \cup O]$ is connected. Let $T$ be a spanning tree of $G[I \cup S \cup O]$. We claim that every node in $S$ is a non-leaf node of $T$. Suppose not, then $G[I \cup S \setminus \{v\}]$ is connected, every vertex in $O$ has a neighbor in $S \setminus \{v\}$, $v$ has a neighbor in $I$ (since $G[I\cup S]$ is connected and $I$ is nonempty), and every vertex in $X \setminus S$ has a neighbor in $I$. Hence $S \setminus \{v\}$ would be an $(I,O)$-extension, contradicting the minimality of $S$. We conclude that every node in $S$ is a non-leaf node of $T$. 
Applying Claim~\ref{obs:internalIndependent} to $S$ in $T$ concludes the proof.
\end{proof}

The next lemma resolves the first subcase, when there are at most $\frac{|L|}{10}$ vertices in $R$ of degree less than $10\ell$. The crucial observation is that in this case, a minimal $(I,O)$-extension~$S$ must be of size significantly smaller than $|L|/2$, due to Lemma~\ref{lem:smallSolution}.

\begin{lemma}\label{lem:boundSmallRsmall}
Let $G$ be a graph, $(L, H, R)$ be an excellent partition of $V(G)$, and $(I, O)$ be a partition of $H \cup R$. If at most $\frac{|L|}{10}$ vertices in $R$ have degree less than $10\ell$ in $G$, then there are at most $2^{|L|} \cdot 2^{-\frac{|L|}{10}}$ minimal $(I,O)$-extensions. Further, the family of all minimal $(I,O)$-extensions can be enumerated in time $2^{|L|} \cdot 2^{-\frac{|L|}{10}} \cdot n^{\Oh(1)}$.
\end{lemma}

\begin{proof}
First, note that $|H| \leq \frac{|L|}{10}$, because $(L, H, R)$ is an excellent partition. Partition $R$ into $R_{\textrm{big}}$ and $R_{\textrm{small}}$ according to the degrees: $R_{\textrm{big}}$ contains all vertices in $R$ of degree at least $10\ell$, while $R_{\textrm{small}}$ contains the vertices in $R$ of degree less than $10\ell$. Since every vertex in $L$ has at most $\ell$ neighbors in $R$, it follows that $|R_{\textrm{big}}| \leq \frac{|L|}{10}$. By assumption $|R_{\textrm{small}}| \leq \frac{|L|}{10}$. It follows that $|R \cup H| \leq \frac{3|L|}{10}$. Now, $I \cup O = R \cup H$, and therefore, by Lemma~\ref{lem:smallSolution} every minimal $(I,O)$-extension has size at most $|I \cup O| \leq \frac{3|L|}{10}$. Hence the number of different minimal $(I,O)$-extensions is at most 
\begin{equation*}
\sum_{i=0}^{\lfloor\frac{3|L|}{10}\rfloor}{|L| \choose  i} \leq 2^{0.882|L|} \leq 2^{|L|} \cdot 2^{-\frac{|L|}{10}}.
\end{equation*}
To enumerate the sets within the given time bound it is sufficient to go through all subsets $S$ of $L$ of size at most $\frac{3|L|}{10}$ and check whether $S$ is a minimal $(I,O)$-extension in polynomial time using the algorithm of Lemma~\ref{lem:polyTimeExtension}.
\end{proof}

We are left with the case when at least $\frac{|L|}{10}$ vertices in $R$ have degree less than $10\ell$ in~$G$.

\begin{lemma}\label{lem:bigRsmallBound}
Let $G$ be a graph, $(L, H, R)$ be an excellent partition of $V(G)$, and $(I, O)$ be a partition of $H \cup R$. If at least $\frac{|L|}{10}$ vertices in $R$ have degree less than $10\ell$ in $G$, then there are at most $2^{|L|} \cdot  e^{-\frac{|L|}{2^{10\ell}\cdot 100\ell^2}}$ minimal $(I,O)$-extensions. The family of all minimal $(I,O)$-extensions can be enumerated in time $2^{|L|} \cdot  e^{-\frac{|L|}{2^{10\ell}\cdot 100\ell^2}} \cdot n^{\Oh(1)}$.
\end{lemma}
\begin{proof}
We assume that $|I\cup O|\geq 2$, since otherwise the claim holds trivially. Let $R_{\textrm{small}}$ be the set of vertices in $R$ of degree less than $10\ell$; by assumption we have $|R_{\textrm{small}}| \geq \frac{|L|}{10}$. Recall that vertices in $R$ have only neighbors in $L$, and every vertex of $L$ has less than $\ell$ neighbors in $R$. Hence, for each vertex $r$ in $R_{\textrm{small}}$ there are at most $10\ell \cdot (\ell-1)$ other vertices in $R_{\textrm{small}}$ that share a common neighbor with $r$. Compute a subset $R'$ of $R_{\textrm{small}}$ as follows. Initially $R'$ is empty and all vertices in $R_{\textrm{small}}$ are unmarked. As long as there is an unmarked vertex $r \in R_{\textrm{small}}$, add $r$ to $R'$ and mark $r'$ as well as all vertices in $R_{\textrm{small}}$ that share a common neighbor with $r$. Terminate when all vertices in $R_{\textrm{small}}$ are marked.

Clearly, no two vertices in the set $R'$ output by the procedure described above can share any common neighbors. Further, for each vertex added to $R'$, at most  $10\ell \cdot (\ell-1) + 1 \leq 10\ell^2$ vertices are marked. Hence, $|R'| \geq \frac{|R_{\textrm{small}}|}{10\ell^2} \geq \frac{|L|}{100\ell^2}$.

Observe that if a subset $S$ of $L$ is an $(I,O)$-extension, then every vertex in $I \cup O$ must have a neighbor in $S$. This holds for every vertex in $O$, because $I \cup S$ needs to dominate this vertex, but there are no edges between $O$ and $I$. For every vertex in $I$ this holds because $G[I \cup S]$ has to be connected, and $I \cup O$ is an independent set of size at least $2$. 

Consider now a subset $S$ of $L$ picked uniformly at random. We upper bound the probability that every vertex in $I \cup O$ has a neighbor in $S$. This probability is upper bounded by the probability that every vertex in $R'$ has a neighbor in $S$. For each vertex $r$ in $R'$, the probability that none of its neighbors is in $S$ is $2^{-d(r)} \geq 2^{-10\ell}$. Since no two vertices in $R'$ share a common neighbor, the events ``$r$ has a neighbor in $S$'' for $r\in R'$ are independent. Therefore, the probability that every vertex in $R'$ has a neighbor in $S$ is upper bounded by
\begin{equation*}
(1-2^{-10\ell})^{|R'|} \leq e^{-2^{-10\ell} \cdot \frac{|L|}{100\ell^2}} = e^{-\frac{|L|}{2^{10\ell}\cdot 100\ell^2}}.
\end{equation*}
The upper bound on the number of minimal $(I,O)$-extensions follows. To enumerate all the minimal $(I,O)$-extensions within the claimed time bound, it is sufficient to enumerate all sets $S \subseteq L$ such that every vertex in $R'$ has at least one neighbor in $S$, and check in polynomial time using Lemma~\ref{lem:polyTimeExtension} whether $S$ is a minimal $(I,O)$-extension. The family of such subsets of $L$ is closed under taking supersets, so to enumerate them we can use the algorithm of Lemma~\ref{lem:enumeration} applied to their complements.
\end{proof}

We can now wrap up the proof of Lemma~\ref{lem:mainBranchingIset}.

\begin{proof}[Proof of Lemma~\ref{lem:mainBranchingIset}]
Let $G$ be a graph and $(L,H,R)$ be a good partition of $V(G)$. Consider a partition of $H \cup R$ into two sets $I$ and $O$. By Lemma~\ref{lem:reduceIO}, we can obtain in polynomial time a graph $G'$ with $V(G)\cap V(G')\supseteq L$, as well as an excellent partition $(L,R',H')$ of $V(G')$, and a partition $(I',O')$ of $R' \cup H'$, such that every subset $S$ of $L$ is a minimal $(I,O)$-extension in $G$ if and only if it is a minimal $(I',O')$-extension in $G'$. Thus, from now on, we may assume without loss of generality that  $L$, $H$ and $R$ is an excellent partition of $V(G)$.

We distinguish between two cases: either there are at most $\frac{|L|}{10}$ vertices in $R$ of degree less than $10\ell$, or there are more than $\frac{|L|}{10}$ such vertices. In the first case, by Lemma~\ref{lem:boundSmallRsmall}, there are at most $2^{|L|} \cdot 2^{-\frac{|L|}{10}}$ minimal $(I,O)$-extensions. Further, the family of all minimal $(I,O)$-extensions can be enumerated in time $2^{|L|} \cdot 2^{-\frac{|L|}{10}} \cdot n^{\Oh(1)}$. In the second case, by Lemma~\ref{lem:bigRsmallBound}, there are at most $2^{|L|} \cdot  e^{-\frac{|L|}{2^{10\ell}\cdot 100\ell^2}}$ minimal $(I,O)$-extensions, and the family of all minimal $(I,O)$-extensions can be enumerated in time $2^{|L|} \cdot  e^{-\frac{|L|}{2^{10\ell}\cdot 100\ell^2}} \cdot n^{\Oh(1)}$. Since  $e^{-\frac{|L|}{2^{10\ell}\cdot 100\ell^2}} \geq 2^{-\frac{|L|}{10}}$, the statement of the lemma follows.
\end{proof}

As argued before, establishing Lemma~\ref{lem:mainBranchingIset} concludes the proof of Lemma~\ref{lem:mainBranching}. We can now use Lemma~\ref{lem:mainBranching} to complete the proof of Lemma~\ref{lem:mainStupid}, and hence also of our main result.

\begin{proof}[Proof of Lemma~\ref{lem:mainStupid}]
To list all minimal connected dominating sets of $G$ it is sufficient to iterate over each of the $2^{n-|L|}$ partitions of $H \cup R$ into $I$ and $O$, for each such partition enumerate all minimal $(I, O)$-extensions $S$ using Lemma~\ref{lem:mainBranching} with $\ell = 14$, and for each minimal extension $S$ check whether $I \cup S$ is a minimal connected dominating set of $G$. Observe that 
\begin{equation*}
|L| \geq \frac{1}{60} \cdot n \geq 10 \cdot 14 \cdot \frac{1}{10^4} \cdot n \geq 10 \cdot 14 \cdot |H|,
\end{equation*}
and that therefore Lemma~\ref{lem:mainBranching} is indeed applicable with $\ell = 14$. Hence, the total number of minimal connected dominating sets in $G$ is upper bounded by
\begin{equation*}
2^{n - |L|} \cdot 2^{|L|} \cdot e^{-\frac{|L|}{2^{10\ell}100\ell^3}} \leq 2^n \cdot 2^{-\frac{n}{60 \cdot 2^{10\ell}100\ell^3}} \leq 2^{n(1-10^{-50})}.
\end{equation*}
The running time bound for the enumeration algorithm follows from the running time bound of the enumeration algorithm of Lemma~\ref{lem:mainBranching} in exactly the same way.
\end{proof}


%



\subparagraph*{Acknowledgements.} The authors gratefully thank the organizers of the workshop {\em{Enumeration Algorithms Using Structure}}, held in Lorentz Center in Leiden in August 2015, where this work was initiated.

\bibliography{references}

\end{document}